\documentclass[oneside,english]{amsart}
\usepackage[T1]{fontenc}
\usepackage[latin9]{inputenc}
\usepackage{prettyref}
\usepackage{bm}
\usepackage{amsthm}
\usepackage{amssymb}

\makeatletter
\numberwithin{equation}{section}
\numberwithin{figure}{section}
\theoremstyle{plain}
\newtheorem{thm}{\protect\theoremname}
\theoremstyle{plain}
\newtheorem{lem}[thm]{\protect\lemmaname}
\theoremstyle{remark}
\newtheorem{rem}[thm]{\protect\remarkname}
\theoremstyle{plain}
\newtheorem{prop}[thm]{\protect\propositionname}
\theoremstyle{plain}
\newtheorem{cor}[thm]{\protect\corollaryname}

\usepackage{mathtools}
\usepackage{braket}
\usepackage{prettyref}

\usepackage{tikz}

\newrefformat{prop}{Proposition \ref{#1}}
\newrefformat{cor}{Corollary \ref{#1}}

\makeatother

\usepackage{babel}
\providecommand{\corollaryname}{Corollary}
\providecommand{\lemmaname}{Lemma}
\providecommand{\propositionname}{Proposition}
\providecommand{\remarkname}{Remark}
\providecommand{\theoremname}{Theorem}

\begin{document}

\title{On the extreme power of nonstandard programming languages}

\author{Takuma Imamura}

\address{Research Institute for Mathematical Sciences\\
Kyoto University\\
Kitashirakawa Oiwake-cho, Sakyo-ku, Kyoto 606-8502, Japan}

\email{timamura@kurims.kyoto-u.ac.jp}
\begin{abstract}
Suenaga and Hasuo introduced a nonstandard programming language ${\bf While}^{{\bf dt}}$
which models hybrid systems. We demonstrate why ${\bf While}^{{\bf dt}}$
is not suitable for modeling actual computations.
\end{abstract}


\maketitle

\section{Introduction}

Suenaga and Hasuo \cite{SH11} introduced an imperative programming
language ${\bf While}^{{\bf dt}}$, which is the usual ${\bf While}$
programming language equipped with a positive infinitesimal ${\bf dt}$.
This language is intended to \emph{hyperdiscretise} hybrid systems
and enable them to be formally verified by Hoare logic \cite{Hoa69}.
On the other hand, this language is not intended to be a model of
\emph{actual} computation. We demonstrate why ${\bf While}^{{\bf dt}}$
is not suitable for modeling actual computations. The main reason
is that ${\bf While}^{{\bf dt}}$ has too strong computational power
\emph{caused by physically impossible settings}. We clarify the causes
of the power.

We refer to Suenaga and Hasuo \cite{SH11} for the definition of ${\bf While}^{{\bf dt}}$;
Robinson \cite{Rob66} for nonstandard analysis; Shen and Vereshchagin
\cite{SV03} for computability theory.

\section{Computation beyond Turing machine model}

\subsection{The first cause: unrestricted use of reals}

The first cause of the power is that ${\bf While}^{{\bf dt}}$ is
furnished with the constant symbols $c_{r}$ for \emph{all} real numbers
$r\in\mathbb{R}$ and the exact comparison operator $<$. They bring
much strong computational power to this language as we will see below.
\begin{lem}
${\bf While}^{{\bf dt}}$ computes the floor function on $\prescript{\ast}{}{\mathbb{R}}$.
\end{lem}

\begin{proof}
The following (pseudo) ${\bf While}^{{\bf dt}}$-program computes
the floor function.
\[
\begin{array}{l}
{\tt Input}:\ x\\
{\tt Output}:\ y\\
n:=0;\\
{\tt while}\ \neg\left(n\leq x<n+1\vee-n\leq x<-n+1\right)\ {\tt do}\\
\quad n:=n+1;\\
{\tt if}\ x\geq0\\
{\tt then}\ y:=n\\
{\tt else}\ y:=-n
\end{array}
\]
\end{proof}
\begin{rem}
The floor function is a typical example of a uncomputable real function
(see Weihrauch \cite{Wei00} p. 6).
\end{rem}

\begin{prop}
\label{prop:WHILEdt-decides-st-probs}${\bf While}^{{\bf dt}}$ computes
every standard decision problem on $\prescript{\ast}{}{\mathbb{N}}$.
\end{prop}

\begin{proof}
Let $A\subseteq\mathbb{N}$. The constant $r=\sum_{i=0}^{\infty}3^{-i}\chi_{A}\left(i\right)$,
where $\chi_{A}$ is the characteristic function of $A$, has complete
information deciding the membership of $A$. Consider the following
program.
\[
\begin{array}{l}
{\tt Input}:\ x\\
{\tt Output}:\ y\\
a:=r;\\
{\tt while}\ x\neq0\ {\tt do}\\
\quad a:=3\cdot a;\\
\quad x:=x-1;\\
{\tt if}\ a-3\cdot\left\lfloor \left(1/3\right)\cdot a\right\rfloor \geq1\\
{\tt then}\ y:=1\\
{\tt else}\ y:=0
\end{array}
\]
This computes the characteristic function $\chi_{\prescript{\ast}{}{A}}$
of $\prescript{\ast}{}{A}$ for all (standard and nonstandard) inputs.
\end{proof}
\begin{cor}
\label{cor:WHILEdt-computes-st-funs}${\bf While}^{{\bf dt}}$ computes
every standard function on $\prescript{\ast}{}{\mathbb{N}}$.
\end{cor}

\begin{proof}
Let $f\colon\mathbb{N}\to\mathbb{N}$. Consider the set $A=\set{J\left(x,f\left(x\right)\right)|x\in\mathbb{N}}$,
where $J\colon\mathbb{N}\times\mathbb{N}\to\mathbb{N}$ is a computable
bijection. ${\bf While^{dt}}$ decides $\prescript{\ast}{}{A}$ by
\prettyref{prop:WHILEdt-decides-st-probs}. The following program
computes $\prescript{\ast}{}{f}$ for all inputs.
\[
\begin{array}{l}
{\tt Input}:\ x\\
{\tt Output}:\ y\\
y:=0;\\
a:=1;\\
{\tt while}\ J\left(x,y\right)\notin\prescript{\ast}{}{A}\ {\tt do}\\
\quad y:=y+1
\end{array}
\]
\end{proof}
\begin{rem}
Here the source of the computational power is not the use of infinitesimals.
The foregoing argument can be applied to any other models of hybrid
systems in which there is no restriction of discrete-continuous interaction.
\end{rem}

\subsection{The second cause: supertasks}

The second cause is that ${\bf While}^{{\bf dt}}$ can execute infinitely
many steps of computation whose computational resource consumption
(such as time, space and electricity usage and heat generation) is
$\gg0$.

While the following result is a special case of \prettyref{cor:WHILEdt-computes-st-funs},
the proof is based on an essentially different idea.
\begin{prop}
${\bf While}^{{\bf dt}}$ computes every ${\bf 0}'$-computable function
on $\mathbb{N}$.
\end{prop}

\begin{proof}
Let $f\colon\mathbb{N}\to\mathbb{N}$ be $\bm{0}'$-computable. By
Schoenfield's limit lemma (see Theorem 48 of \cite{SV03}), there
is a computable function $F\colon\mathbb{N}\times\mathbb{N}\to\mathbb{N}$
such that $f=\lim_{s\to\infty}F\left(s,-\right)$. Obviously ${\bf While}^{{\bf dt}}$
computes $F$ for all inputs (with no use of uncomputable real numbers).
Consider the following program:
\[
\begin{array}{l}
{\tt Input}:\ x\\
{\tt Output}:\ y\\
y:=F\left(\infty,x\right)
\end{array}
\]
This computes the limit function $f$ for all standard inputs.
\end{proof}
\begin{rem}
The infinity constant $\infty$ can be eliminated as follows.
\[
\begin{array}{l}
t:=0;\\
u:=0;\\
{\tt while}\ t<1\ {\tt do}\\
\quad t:=t+{\tt dt};\\
\quad u:=u+1
\end{array}
\]
The variable $u$ is infinite after executing this program. The while
loop is repeated an infinite number of times. The instruction $u:=u+1$
in the loop consumes computational resource $\gg0$ in each execution.
\end{rem}

If ${\bf While^{dt}}$ can finish such an infinite sequence of operations
\emph{only within infinite time}, there is no problem involving the
computational power, because the computation by ${\bf While^{dt}}$-programs
is not considered to be actual one. On the other hand, if we want
to consider ${\bf While^{dt}}$ to be a model of computation in the
real world, ${\bf While^{dt}}$ must be able to finish such an infinite
sequence of operations \emph{within finite time}, or, in other words,
must admit \emph{supertasks}.

Thomson \cite{Tho54} did an insightful thought experiment to analyse
the (im)possibility of supertasks. There is a lamp with a switch.
The initial state of the switch is off. Consider the following supertask:
in the first $1/2$ sec, turn on the switch; in the next $1/4$ sec,
turn off the switch; in the next $1/8$ sec, turn on the switch; and
so on. After this supertask, is the lamp on or off? A similar circumstance
occurs in ${\bf While^{dt}}$. We identify `on' with $1$ and `off'
with $0$. Consider the following ${\bf While^{dt}}$-program. 
\[
\begin{array}{l}
{\it time}:=0;\\
{\it lamp}:=0;\\
{\tt while}\ {\it time}<1\ {\tt do}\\
\quad{\it time}:={\it time}+{\tt dt};\\
\quad{\it lamp}:=1-{\it lamp}
\end{array}
\]
This program eventually halts. The same question then arises: after
execution, is the value of ${\it lamp}$ on or off? The answer depends
on the interpretation of the infinitesimal constant ${\tt dt}$.
\begin{rem}
The same phenomena occur in other models of hypercomputation which
admit supertasks, such as the accelerated Turing machines (Copeland
\cite{Cop02}; Calude and Staiger \cite{CS10}).
\end{rem}

\section{Conclusion}

The nonstandard programming language ${\bf While}^{{\bf dt}}$ has
too strong computational power. However, this computational power
\emph{per se} is not an essential reason why this model is inappropriate,
because the Church--Turing thesis may be false (i.e. hypercomputation
may be physically realisable). The excessive power is a consequence
of the following causes: the unrestricted use of reals, the exact
comparison of reals, and the possibility of supertasks consuming infinite
resources. These are physically impossible. Because of this impossibility,
while actual hybrid systems can be modeled by ${\bf While^{dt}}$-programs,
some ${\bf While^{dt}}$-programs do not represent any actual hybrid
system. The same applies to other nonstandard ``models of computation''
such as ${\bf Sproc}^{{\bf dt}}$ (Suenaga \emph{et.al.} \cite{SSH13}),
${\bf NSF}$ (Nakamura \emph{et.al.} \cite{NKSI17,Nak18}) and the
internal Turing machines (Loo \cite{Loo04}).

Some restrictions are needed to metamorphose ${\bf While}^{{\bf dt}}$
into a model of actual hybrid computation. For instance, restricting
the electricity usage and/or the heat generation to finite, one can
avoid ``Thomson-type'' problems. In the Thomson's lamp experiment,
one needs infinite energy to switch the lamp infinitely many times.
This is an example of a \emph{bad} supertask. On the other hand, a
rubber ball uses only finite energy (the initial mechanical energy)
to bounce infinitely many times (\prettyref{fig:rubber-ball}). This
is an example of a \emph{good} supertask.
\begin{figure}
\begin{tikzpicture}
	\fill (0, 4) circle[radius=3pt];
	\draw[scale=1,domain=0:1,smooth,variable=\x] plot ({\x},{4 - 4 * \x * \x});
	\draw[scale=1,domain=1:2,smooth,variable=\x] plot ({\x}, {- 8 * (\x - 1) * (\x - 2)});
	\draw[scale=1,domain=2:2.5,smooth,variable=\x] plot ({\x}, {- 16 * (\x - 2) * (\x - 2.5)});
	\draw[scale=1,domain=2.5:2.625,smooth,variable=\x] plot ({\x}, {- 32 * (\x - 2.5) * (\x - 2.75)});
	\draw (-1,0) -- (4, 0);
\end{tikzpicture}\caption{\label{fig:rubber-ball}The rubber ball loses some kinetic energy
in each inelastic collision. The kinetic energy lost is converted
into other forms such as thermal energy. Then the interval of collision
becomes shorter. Thus the ball bounces infinitely many times within
finite time.}
\end{figure}
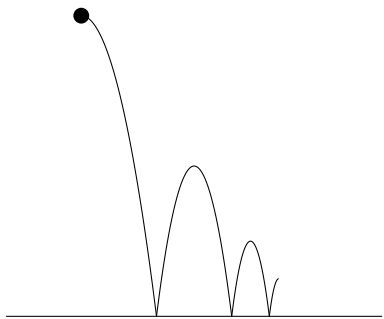

\bibliographystyle{plain}
\bibliography{bibtex}

\end{document}